\documentclass[letterpaper, 10 pt, conference]{ieeeconf}

\IEEEoverridecommandlockouts
\usepackage{graphicx}
\usepackage{epsfig}
\usepackage{amsmath}

\usepackage{amsthm}

\usepackage{amssymb}
\usepackage[noadjust]{cite}
\usepackage{bm}
\usepackage{subfigure}
\usepackage{acronym}
\usepackage{paralist}
\usepackage{float}
\usepackage{epstopdf}
\usepackage{algorithm}
\usepackage{algpseudocode}
\usepackage{mathtools}
\usepackage{multicol} 
\usepackage{accents}
\newcommand{\subparagraph}{}
\usepackage{titlesec}
\usepackage{xcolor}

\newtheorem{define}{Definition}
\newtheorem{theorem}{Theorem}
\newtheorem{lemma}{Lemma}

\newtheorem{remark}{Remark}
\newtheorem{assume}{Assumption}

\newcommand{\Rc}{\mathcal R}
\newcommand{\R}{\mathbb R}

\newcommand{\N}{\mathcal{N}}
\newcommand{\D}{\mathcal{D}}
\newcommand{\V}{\mathcal{V}}

\newcommand{\G}{\mathcal{G}}
\newcommand{\E}{\mathcal{E}}
\newcommand{\Fc}{\mathcal{F}}

\newcommand{\Le}{\mathcal{L}}
\newcommand{\A}{\mathcal{A}}
\newcommand{\J}{\mathcal{J}}
\newcommand{\eps}{\epsilon}

\newcommand{\Z}{\mathbb{Z}}

\newcommand{\Pc}{\mathcal{P}}

\newcommand{\bmx}[1]{\begin{bmatrix}#1\end{bmatrix}} 
\newcommand{\pth}[1]{\left(#1\right)} 
\newcommand{\brc}[1]{\left \{#1\right \}} 
\newcommand{\nrm}[1]{\left \lVert#1\right \rVert} 

\newcommand{\bmxs}[1]{\begin{bsmallmatrix}#1\end{bsmallmatrix}}

\newcommand{\ggeq}{\succeq} 

\DeclarePairedDelimiter{\ceil}{\lceil}{\rceil}
\DeclarePairedDelimiter{\floor}{\lfloor}{\rfloor}
\DeclarePairedDelimiter{\abs}{\lvert}{\rvert}

\newcommand{\rarr}{\rightarrow} 

\makeatletter
\let\oldceil\ceil
\def\ceil{\@ifstar{\oldceil}{\oldceil*}}
\makeatother
\makeatletter
\let\oldfloor\floor
\def\floor{\@ifstar{\oldfloor}{\oldfloor*}}
\makeatother
\makeatletter
\let\oldnorm\norm
\def\norm{\@ifstar{\oldnorm}{\oldnorm*}}
\makeatother
\makeatletter
\let\oldabs\abs
\def\abs{\@ifstar{\oldabs}{\oldabs*}}
\makeatother

\newcommand{\limfy}[1]{\lim_{#1 \rarr \infty}}
\newcommand{\sign}{\text{sign}}
\newcommand{\cvx}{\text{co}}
\newcommand{\cvxc}{\overline{\text{co}}}

\newcommand{\xN}{x_{\N}}
\newcommand{\Lt}{\widetilde{\mathcal{L}}}
\newcommand{\fN}{f_{\N}}
\newcommand{\RN}{\R^{|\N|}}
\newcommand{\rv}{\color{black}}

\newtheorem{problem}{Problem}
\newtheorem{proposition}{Proposition}

\title{\LARGE \bf
Resilient Finite-Time Consensus: A Discontinuous Systems Perspective
}

\author{James Usevitch and Dimitra Panagou
\thanks{James Usevitch and Dimitra Panagou are with the Department of Aerospace Engineering, University of Michigan, Ann Arbor, MI 48109, USA.
        {\tt\small \{usevitch, dpanagou\}@umich.edu}. The authors would like to acknowledge the support of the Automotive Research Center (ARC) in accordance with Cooperative Agreement W56HZV-14-2-0001 U.S. Army TARDEC in Warren, MI, and of the Award No W911NF-17-1-0526.}%
}

\begin{document}

\maketitle
\thispagestyle{empty}
\pagestyle{empty}

\begin{abstract}

Many algorithms have been proposed in prior literature to guarantee resilient multi-agent consensus in the presence of adversarial attacks or faults. The majority of prior work present excellent results that focus on discrete-time or discretized continuous-time systems. Fewer authors have explored applying similar resilient techniques to continuous-time systems without discretization. 
These prior works typically consider asymptotic convergence and make assumptions such as continuity of adversarial signals, the existence of a dwell time between switching instances for the system dynamics, or the existence of trusted agents that do not misbehave.
In this paper, we expand the study of resilient continuous-time systems by removing many of these assumptions and {\rv using discontinuous systems theory to provide conditions for normally-behaving agents with nonlinear dynamics to achieve consensus in finite time despite the presence of adversarial agents.}
\end{abstract}

\section{Introduction}

Recent years have seen increased interest in multi-agent control systems that can accomplish desired control objectives despite the presence of adversarial or faulty agents. In particular, the resilient consensus problem has been the focus of much attention. Many approaches for discrete-time systems based on the family of Mean-Subsequence-Reduced (MSR) algorithms have been developed to ensure that normally-behaving agents in a multi-agent system can achieve consensus despite a bounded number of arbitrarily misbehaving agents \cite{leblanc2013resilient, dibaji2019resilient, dibaji2017quantized, dibaji2017resilient, mitra2019byzantine, zhang2012simple, saldana2017resilient}. MSR algorithms typically operate by having agents update their states with a trimmed mean of the local values received from their in-neighbors. Additional conditions on the network structure and the scope of the adversarial threat guarantee consensus of the normally-behaving agents. The majority of papers using MSR-based algorithms consider either discrete-time systems or discretized continuous systems.

Less attention has been devoted to studying counterparts of these MSR algorithms designed for continuous-time systems that are not discretized \cite{leblanc2013continuous, leblanc2017resilient, oksuz2019resilient, shang2019consensus, wu2017secure}. 
One of the difficulties in studying resilient consensus in the continuous-time domain with arbitrarily misbehaving adversaries is the issue of existence and uniqueness of system solutions that describe normal agents' state trajectories. 
{\rv
For example, guaranteeing existence and uniqueness of system solutions can become difficult when adversarial signals are discontinuous without a minimum dwell time between discontinuities.
}
In the seminal work \cite{leblanc2013continuous} the \emph{Adversarial Robust Consensus Protocol} (ARC-P) was presented, where continuous-time single-integrator agents apply a trimmed-mean approach to achieve resilient consensus. These results were extended in \cite{leblanc2017resilient} to more general LTI agents achieving state synchronization. A limiting assumption made in \cite{leblanc2013continuous,leblanc2017resilient} is that all signals sent from adversarial agents to normal agents are continuous in time. 
The authors of \cite{leblanc2013continuous} give reasonable justifications for this assumption, but their results have not yet been extended to more general adversarial signals that may exhibit discontinuities. 
A few prior works have made the assumption of minimum dwell time between instances where the system dynamics change due to filtering \cite{shang2019consensus, wu2017secure}. Nevertheless for many prior control algorithms it is possible to construct adversarial signals that cause infinite switching of system dynamics in a finite amount of time (which is demonstrated in Section \ref{sec:nonsmoothreview} of this paper). The works \cite{shang2019consensus, wu2017secure, oksuz2019resilient} do not discuss the possibility of discontinuous adversarial signals or the existence and uniqueness of system solutions.

Finite-time consensus is also a current topic of interest in the literature \cite{bhat2000finite, usevitch2018finite, garg2018new, liu2017finite}. Much of the prior literature on finite-time consensus assumes all agents apply the nominally specified controllers. There is relatively little work that treats finite-time consensus in the presence of adversarial agents. Some examples include the excellent results in \cite{franceschelli2014finite,franceschelli2017finite}. However, in \cite{franceschelli2014finite} it is assumed that only the \emph{initial conditions} of certain agents are faulty, with all agents applying the nominally specified control protocol. In contrast, Byzantine adversaries may apply an arbitrary control protocol at any instant subsequent to the initial time. In addition, \cite{franceschelli2017finite} considers only undirected graphs, assumes that there exists a safe set of trusted agents that never misbehave, and assumes that all misbehaving agents are only connected to trusted agents.

This work approaches the problem of resilient continuous-time consensus from a \emph{discontinuous systems} perspective \cite{cortes2008discontinuous} and relaxes many of the assumptions of prior literature. 
We present a novel nonlinear resilient control algorithm and conditions under which normally-behaving agents achieve consensus in finite time despite the presence of misbehaving adversarial agents.
The contributions of this paper can be summarized as follows:
\begin{itemize}
	\item We introduce a novel controller that guarantees finite-time consensus for a class of nonlinear systems in the presence of adversarial attacks and faults.
	\item We demonstrate using discontinuous systems theory that our analysis holds even for discontinuous adversarial signals with no minimum dwell time between discontinuities.
	\item We demonstrate that our analysis holds for the general $F$-local adversarial model on digraphs, which does not assume the presence of any trusted agents.
\end{itemize}

This paper is organized as follows: Section \ref{sec:notation} introduces the notation and problem formulation, Section \ref{sec:main} presents our main results, Section \ref{sec:simulations} gives simulations demonstrating our method, and Section \ref{sec:conclusion} gives a brief conclusion.

\section{Notation and Problem Formulation}
\label{sec:notation}

The sets of real numbers and integers are denoted $\R$ and $\Z$, respectively. The sets of nonnegative real numbers and integers are denoted $\R_+$ and $\Z_+$, respectively. 
The cardinality of a set $S$ is denoted as $|S|$, and the empty set is denoted $\emptyset$. 
The power set is denoted as $\Pc(S)$.
The convex hull of a set $S$ is denoted $\cvx\{S\}$, and the convex closure of a set $S$ is denoted $\cvxc\{S\}$.
{\rv The notations $B(x,\eps)$, $\bar{B}(x,\eps)$ denote the open and closed balls of radius $\eps > 0$ at $x \in \R^d$, respectively.}
The notations $\bm 1$ and $\bm 0$ denote the vector of all ones and the vector of all zeros, respectively, where the size of the vectors will be implied by the context.
The $i$th column of the identity matrix $I$ is denoted $  e^i$, with $I = \bmxs{ e^1 &  e^2 & \ldots &  e^n}$. 
A directed graph (digraph) is denoted as $\D = (\V,\E)$, where $\V = \{1,\ldots,n\}$ is the set of indexed vertices and $\E$ is the edge set. 
A directed edge is denoted $(i,j)$, with $i,j \in \V$, meaning that agent $j$ can receive information from agent $i$. The set of in-neighbors for an agent $j$ is denoted $\V_j = \{i \in \V : (i,j) \in \E \}$. The set of inclusive in-neighbors is defined as $\J_i = \V_i \cup \{i\}$. The set of out-neighbors for an agent $j$ is denoted $\V_j^{\text{out}} = \{i \in \V : (j,i) \in \E\}$. 
The sign function ($\sign : \R \rarr \R$) is defined as follows:
\begin{align}
	\sign(x) = \begin{cases}
		1 & \text{if } x > 0\\
		0 & \text{if } x = 0\\
		-1 & \text{if } x < 0
	\end{cases}, \hspace{1em} x \in \R
\end{align}

The notions of $r$-reachability and $r$-robustness will be used in this paper to quantify the graph theoretic conditions guaranteeding resilient consensus:
\begin{define}[\cite{leblanc2013resilient}]
\label{def:rreach}
Let $r \in \Z_+$ and $\D=(\V,\E)$ be a digraph. A nonempty subset $S \subset \V$ is $r$-reachable if $\exists i \in S$ such that $|\N_i \backslash S| \geq r$.
\end{define}
\begin{define}[\cite{leblanc2013resilient}]
\label{def:rrobust}
Let $r \in \Z_+$. A nonempty, nontrivial digraph $\D = (\V,\E)$ on $n$ nodes $(n \geq 2)$ is $r$-robust if for every pair of nonempty, disjoint subsets of $\V$, at least one of the subsets is $r$-reachable.
\end{define}

\subsection{Problem Formulation}
\label{sec:probform}

Consider a network of $n$ agents with $n \geq 2$ whose communication structure is modeled by the digraph $\D = (\V, \E)$. Without loss of generality we assume an initial time of $t_0 = 0$.
Each agent $i$ has a scalar state $x_i : \R \rarr \R$ and continuous-time first-order dynamics
\begin{align}
\label{eq:dynamics}
	\dot{x}_i(t) = u_i(t)
\end{align}
where the form of $u_i(t)$ will be given in Algorithm \ref{alg:FTRC}. 
At all times $t \geq 0$ each agent $i$ is able to send a signal to its out-neighbors containing a function of its state $g(x_i(t))$, where $g : \R \rarr \R$ is a strictly increasing function with domain equal to $\R$. The function $g(\cdot)$ is the same for all agents and is not required to be continuous.


\begin{define}
	The notation $g(x_j^i(t))$, $x_j^i : \R \rarr \R$, denotes the signal received by agent $i$ from agent $j$ at time $t$.
\end{define}

A \emph{normally-behaving} agent is defined as an agent $i$ that sends the function of its true state value $g(x_i(t))$ to all of its out-neighbors and updates its state according to the \emph{Finite-Time Resilient Consensus Protocol} (FTRC-P) defined in Algorithm \ref{alg:FTRC}. The set of all normal agents is denoted $\N \subset \V$.

\begin{algorithm}
\caption{\small{\textsc{FTRC Protocol (FTRC-P)}:}}
\label{alg:FTRC}
\begin{enumerate}
	\item At time $t$, each normal agent $i$ receives values $g(x_j^i(t))$ from its in-neighbors $j \in \V_i(t)$ and forms a sorted list.
	\item If there are less than $F$ values strictly larger than $i$'s own value $g(x_i(t))$, then $i$ removes all values which are strictly larger than its own. Otherwise $i$ removes precisely the largest $F$ values in the sorted list.
	\item In addition, if there are less than $F$ values strictly smaller than $i$'s own value $g(x_i(t))$, then $i$ removes all values which are strictly smaller than its own. Otherwise $i$ removes precisely the smallest $F$ values in the sorted list.
	\item Let $\Rc_i(t)$ denote the set of agents whose values are removed by agent $i$ in steps 2) and 3) at time $t$. Agent $i$ applies the following update:
	\begin{align}
\label{eq:uFTRC}
	u_i(t) = \alpha\, \sign\pth{\sum_{\J_i \backslash \Rc_i[t]} g(x_j^i(t)) - g(x_i(t)) }
\end{align}
where $\alpha > 0$ and $g: \R \rarr \R$ is defined in Section \ref{sec:probform}. Note that since $i \in \J_i$ by definition and agent $i$ never filters out the function of its own state $g(x_i(t))$, \eqref{eq:uFTRC} is always well-defined.
\end{enumerate}
\end{algorithm}


We consider the presence of misbehaving adversaries in this problem setting, which are defined as follows:
\begin{define}
\label{def:misbehaving}
	An agent $k \in \V$ is called \emph{misbehaving} if at least one of the following conditions holds:
	\begin{itemize}
		\item There exists $t \geq t_0$ such that $u_k(t)$ is not equal to the input \eqref{eq:uFTRC} defined by the FTRC Protocol in Algorithm \ref{alg:FTRC}.
		\item There exists $i \in \V_k^{\text{out}}$ and $t \geq t_0$ such that $g(x_k^i(t)) \neq g(x_k(t))$; i.e. agent $k$ sends an out-neighbor a different value than its actual state value.
		\item There exists $i_1, i_2 \in \V_k^{\text{out}}$ and $t \geq t_0$ such that $g(x_k^{i_1}(t)) \neq g(x_k^{i_2}(t))$; i.e. agent $k$ sends different values to different out-neighbors.
	\end{itemize}
The set of misbehaving agents is denoted $\A \subset \V$.
\end{define}

Note that the definition of misbehaving agents encompasses both \emph{Byzantine} adversaries \cite{leblanc2013continuous} and faulty agents.
All nodes in $\V$ are either normal or misbehaving; i.e. $\A \cap \N = \emptyset$ and $\A \cup \N = \V$.
{\rv 
The only assumption made on the signals $g(x_k^i(\cdot))$ originating from the adversaries is the following condition:
\begin{assume}
\label{assume:measurable}
	For any $k \in \A$ and $i \in \N$, the function $g \circ x_k^i$ is Lebesgue measurable.
\end{assume}

\begin{remark}
	Assumption \ref{assume:measurable} widens the class of adversarial signals that can be considered as compared to prior work. Prior work typically assumes that adversarial signals are continuous \cite{leblanc2013continuous, leblanc2017resilient} or have a finite number of discontinuities in any compact interval \cite{shang2019consensus, wu2017secure}. Under Assumption \ref{assume:measurable} however, the techniques in this paper consider adversarial signals which may be discontinuous and have possibly infinite discontinuities in a finite interval.

Naturally, Assumption \ref{assume:measurable} raises the question of what happens if one or more of the adversarial signals are not Lebesgue measurable. The answer to this question hinges upon whether there exist subsets of $\R$ which are not Lebesgue measurable, which in itself depends on which core axioms of mathematics are assumed to hold (e.g. the axiom of choice). Further discussion on this point is given in the Appendix in Section \ref{sec:discussionassume}.
\end{remark}
}

To quantify the number and distribution of misbehaving agents in the network, we will use the $F$-local model commonly employed in prior literature.
\begin{define}[\cite{leblanc2013resilient}]
A set $S \subset \V$ is $F$-local for $F \in \Z_+$ if it contains at most $F$ nodes in the neighborhood of the other nodes for all $t \geq 0$; i.e. $|\V_i \cap S| \leq F$ $\forall i \in \V \backslash S$, $\forall t \geq 0$.
\end{define}

Note that under the $F$-local model, no agents are assumed to be trusted, i.e. invulnerable to attacks or faults.


The objective of the normal agents is to achieve consensus in their state values despite the presence of an $F$-local adversarial set $\A$. We ultimately are not concerned with the trajectories of the adversarial agents' states--we are only concerned with ensuring that the actions of the adversarial agents do not prevent the consensus of the normal agents.
In this light, we define the vector of normal agents' states as follows:
\begin{align}
\label{eq:normalvector}
	\xN(t) = \bmx{x_{\N_1}(t) \\ x_{\N_2}(t) \\ \vdots \\ x_{\N_{|\N|}}(t)},\ \xN(t) \in \R^{|\N|},
\end{align}
where $\N_j$ is the index of the $j$th agent in $\N$ according to any arbitrary fixed ordering of $\N$, with $\{\N_1,\N_2, \ldots, \N_{|\N|} \} = \N$.
To give a brief example, in a network of $n = 5$ agents with the normal agents being $\{2,4,5\}$, we have $\N_1 = 2$, $\N_2 = 4$, and $\N_3 = 5$ with $\xN(t) = \bmxs{x_2(t) & x_4(t) & x_5(t)}^T$.
Consensus of the normal agents is achieved when $\xN(t) \in \text{span}(\bm 1)$. 
However, note by the form of \eqref{eq:uFTRC} that each $u_i(\cdot)$ is a function of both signals from normal agents \emph{and} signals from any adversarial agents that are in-neighbors of $i$. For all $i \in \N$, the vector of adversarial signals sent to $i$ at time $t$ is denoted $x_{\A}^{i} \in \R^{|\V_i \cap \A|}$. 
The dynamics of the normal agents are therefore written as follows:
\begin{align}
	\dot{ x}_{\N}(t) &= \bmx{u_{\N_1}( x_{\N}(t),  x_{\A}^{\N_1}(t))) \\ u_{\N_2}( x_{\N}(t),  x_{\A}^{\N_2}(t)) \\ \vdots \\ u_{\N_{|\N|}}( x_{\N}(t), x_{\A}^{\N_{|\N|}}(t))},  \nonumber\\
	&= f_{\N}(x_{\N}(t), x_{\A}^{\N}(t)), \label{eq:totalsystem}
\end{align}
where  $\{\N_1, \ldots, \N_{|\N|}\} = \N$ and
\begin{align}
\label{eq:alladversaries}
{
\medmuskip=0mu
\thinmuskip=0mu
\thickmuskip=0mu
x_{\A}^{\N}(t) = \bmx{(x_{\A}^{\N_1}(t))^T & \hspace{-.5em} \cdots \hspace{-.5em} & (x_{\A}^{\N_{|\N|}}(t))^T}^T \in \R^{\sum_{\N_j \in \N} |\V_{\N_j} \cap \A|}
}
\end{align}
is the vector of all adversarial signals at time $t$.
	By definition, the adversarial signals are arbitrary functions of time and in general will not be functions of the normal agent state vector $\xN(t)$. The adversarial signals in each vector $ x_{\A}^{\N_i}$ can therefore be viewed as arbitrary, possibly discontinuous inputs to the system of normal agents.


The objective of the normally-behaving agents is to achieve
\emph{Finite-Time Resilient Consensus} (FTRC). To define FTRC, we first introduce the following functions: 
\begin{align}
	M(\xN) &= \max_{i \in \N} x_i = \max_{j \in \{1,\ldots,|\N|\}}(e^j)^T \xN \nonumber \\
	m(\xN) &= \min_{i \in \N} x_i = \min_{j \in \{1,\ldots,|\N|\}}(e^j)^T \xN \nonumber \\
	V(\xN) &= M(\xN) - m(\xN) \label{eq:Lyapunov}
\end{align}
We also define the following sets to describe the agents with state values equal to $M(\xN)$ or $m(\xN)$:
\begin{align}
S_M &= \{i \in \N : x_i = M(\xN) \} \nonumber \\
S_m &= \{i \in \N : x_i = m(\xN) \} \label{eq:Smsets}
\end{align}


%
\begin{define}
\label{def:FTRC}
	The normal agents $i \in \N$ achieve Finite-Time Resilient Consensus (FTRC) if all of the following conditions hold:
	\begin{itemize}
		\item[(i)] $x_i(t) \in [m(\xN(0)), M(\xN(0))]$ for all $t \geq 0$ and for all $i \in \N$.
		\item[(ii)] $\exists T : \R^{|\N|} \rarr \R_+$ such that $V(\xN(t)) = 0$ for all $t \geq T(x_{\N}(0))$. {\rv Equivalently, $\xN(t) \in \text{span}(\bm 1)$ for all $t \geq T(\xN(0))$.}
	\end{itemize}
\end{define}
\begin{remark}
The notion of FTRC is based on the notion of \emph{Continuous-Time Resilient Asymptotic Consensus} (CTRAC) in \cite{leblanc2013continuous}, but imposes the stricter requirement that $V(\xN(t))$ converges \emph{exactly} to zero in a finite amount of time and remains there for all future time.
\end{remark}



\begin{problem}
\label{prob:FTRC}
	Determine conditions under which FTRC is achieved by the normal agents $i \in \N$ in the presence of a misbehaving subset of agents $\A \subset \V$.
\end{problem}


%
%
%

%

\subsection{Justification for Discontinuous Systems Approach}
\label{sec:nonsmoothreview}

This paper uses discontinuous systems theory and nonsmooth analysis to prove that a network of agents applying the FTRC-P achieves FTRC. There are two reasons for such an approach. First, the form of $u_i(\cdot)$ in \eqref{eq:uFTRC} implies that the right hand side (RHS) of \eqref{eq:totalsystem} is discontinuous. Note that we cannot simply assume a minimum ``dwell time" and treat the system as a switching system, since cleverly designed adversarial signals may induce an arbitrary number of discontinuities in any given time interval. To give a pathological example, suppose an agent $i \in \N$ receives an adversarial signal $x_k^i(t)$ from $k \in \A$ defined as follows:
\begin{align}
x_k^i(t) = \begin{cases}
a \in \R & \text{if } t \in \mathbb{I},\\
b \in \R,\ b \neq a & \text{if } t \in \mathbb{Q}
\end{cases}
\end{align}
where $a$ and $b$ are chosen appropriately, and $\mathbb{I}$ and $\mathbb{Q}$ represent the sets of irrational and rational numbers in $\R$, respectively. Both $\mathbb{I}$ and $\mathbb{Q}$ are dense in $\R$, implying that no positive minimum dwell time can be assumed for the system. 
The second reason for a discontinuous systems approach is that the Lyapunov-like candidate $V(\xN(t))$ from \eqref{eq:Lyapunov} which will be used for convergence analysis is nonsmooth in general. Discontinuous systems theory allows for nonsmoothness and discontinuities to be addressed in a mathematically precise manner while solving Problem \ref{prob:FTRC}.

\subsection{Review of Discontinuous Systems Theory}
\label{sec:Review}
This subsection gives a brief overview of several fundamental concepts from discontinuous systems theory that are relevant to this paper. The reader is referred to \cite{cortes2008discontinuous, clarke1990optimization, shevitz1994lyapunov, bacciotti2004nonsmooth} for more detailed information.

A differential inclusion is a system with dynamics
\begin{align}
\dot{x}(t) \in \mathcal{F}(t,x(t)), \label{eq:examplediffincl}
\end{align}
where $x : \R \rarr \R^d$ and $\Fc : \R^d \rarr \Pc(\R^d)$, where $\Pc(\R^d)$ denotes the power set of $\R^d$ as defined in Section \ref{sec:notation}. The set-valued map $\Fc$ indicates that at every time $t$ there can be multiple possible evolutions of the system state rather than just one. 
A Caratheodory solution of \eqref{eq:examplediffincl} defined on $[t_0,t_1] \subset [0,\infty)$ is an absolutely continuous function $x : [t_0,t_1] \rarr \R^d$ such that $\dot{x}(t) \in \Fc(t,x(t)$ for almost all $t \in [t_0,t_1]$ in the sense of Lebesgue measure. Existence of Caratheodory solutions to \eqref{eq:examplediffincl} is guaranteed by the following proposition:

\begin{proposition}[\cite{cortes2008discontinuous}]
\label{prop:S2}
Suppose the set-valued map $\mathcal{F} : [0,\infty) \times \R^d \rarr \Pc(\R^d)$ is locally bounded and takes nonempty, compact and convex values. Assume that, for each $t \in \R$, the set-valued map $x \mapsto \mathcal{F}(t,x)$ is upper semicontinuous, and for each $x \in \R^d$, the set-valued map $t \mapsto \mathcal{F}(t,x)$ is measurable. Then, for all $(t_0, x_0) \in [0,\infty) \times \R^d$ there exists a Caratheodory solution of \eqref{eq:examplediffincl} with initial condition $x(t_0) = x_0$.
\end{proposition}

For convenience, the definitions of locally bounded, upper semicontinuity, and local Lipschitzness are given below.
\begin{define}[Locally bounded \cite{cortes2008discontinuous}]
	The set-valued map $\mathcal{F} : [t_0,\infty) \times \R^d \rarr \mathcal{P}(\R^d)$ is locally bounded at $(t,   x) \in [t_0, \infty) \times \R^d$ if there exist $\eps,\delta > 0$ and an integrable function $m : [t,t+\delta] \rarr (0,\infty)$ such that $\nrm{z}_2 \leq m(s)$ for all $z \in \mathcal{F}(s,  y)$, all $s \in [t, t+\delta]$, and all $  y \in B(  x, \eps)$ where $B(x,\eps)$ is the unit ball of radius $\eps$ centered at $x$.
\end{define}

\begin{define}[Upper semicontinuity \cite{cortes2008discontinuous}]
The time-invariant set-valued map $\mathcal{F} : \R^d \rarr \Pc(\R^d)$ is upper semicontinuous at $  x \in \R^d$ if for all $\eps > 0$ there exists $\delta > 0$ such that $\mathcal{F}(  y) \subseteq \mathcal{F}(  x) + B(  0,\eps)$ for all $y \in B(  x, \delta)$.
\end{define}

%
{\rv
\begin{define}[\cite{cortes2008discontinuous}]
The set-valued map $\mathcal{F} : [t_0,\infty) \times \R^d \rarr \mathcal{P}(\R^d)$ is locally Lipschitz at $x \in \R^d$ if there exists $L(x), \eps > 0$ such that $\mathcal{F}(y) \subset \mathcal{F}(z) + L(x)\nrm{y - z}_2 \bar{B}(0,1)$ for all $y,z \in B(x,\eps)$. Note that a set-valued map being locally Lipschitz implies that it is also upper semi-continuous \cite{cortes2008discontinuous}.
\end{define}
}

Existence intervals for Caratheodory solutions to \eqref{eq:examplediffincl} can be extended forward in time using the following result.

\begin{theorem}[\cite{filippov2013differential} Ch. 2 \S 7 Thm 4]
\label{thm:Filippov}
	Let $\Fc : \R^d \rarr \Pc(\R^d)$ satisfy the hypotheses of Proposition \ref{prop:S2} in a compact domain $D  \subset \R \times \R^d$, and be upper semicontinuous in $t$ and $x$ on $D$. Then each solution of \eqref{eq:examplediffincl} with $\bmxs{t_0 \\ x(t_0)} \in D$ can be continued in time until $\bmxs{t \\ x(t)}$ reaches the boundary of $D$.
\end{theorem}

Although there are multiple ways to define set-valued maps, the following method will be used in this paper.
\begin{define}[\cite{cortes2008discontinuous}]
\label{def:diffincl}
	Let $f : \R^d \times \mathcal{U} \rarr \R$, where $\mathcal{U} \subset \R^m$ is the set of allowable control inputs, and let $u : \R \rarr \mathcal{U}$ be a control signal. Consider the function
$\dot{ x}(t) = f( x(t),   u(t)),\   u(t) \in \mathcal{U}$.
	The set-valued map $G[  f] : \R^d \rarr \Pc(\R^d)$ is defined as
	\begin{align}
		G[f](  x) \triangleq \brc{f(  x,  u) :   u \in \mathcal{U}}.
	\end{align}
\end{define}

The notion of generalized gradient extends the notion of gradient to locally Lipschitz functions that may not be continuously differentiable everywhere.

\begin{define}[Generalized Gradient \cite{clarke1990optimization, shevitz1994lyapunov}]
	Let $V : \R^d \rarr \R$ be a locally Lipschitz function \cite[Sec. 3.1]{khalil2002nonlinear}, and let $\Omega_V \subset \R^d$ denote the set of points where $V$ fails to be differentiable,\footnote{Note that by Rademacher's Theorem, a locally Lipschitz function is differentiable almost everywhere in the sense of Lebesgue measure \cite[Sec. 1.2]{clarke1990optimization}.} and let $S \subset \R^d$ denote any other set of measure zero. The \emph{generalized gradient} $\partial V : \R^d \rarr \Pc(\R^d)$ of $V$ is defined as
	\begin{align}
		\partial V(  x) = \cvx \brc{\limfy{i} \nabla V(  x^i) :   x^i \rarr   x,\   x^i \notin \Omega_V \cup S}	
	\end{align}
\end{define}

Computing generalized gradients can be difficult in general. However several useful results exist in the literature that facilitate this calculation, including the following one.

\begin{proposition}[\cite{cortes2008discontinuous}]
\label{prop:maxfunc}
	For $k \in \{1,\ldots,m \}$, let $g_k : \R^d \rarr \R$ be locally Lipschitz at $x \in \R^d$, and define the functions $g_{\max} : \R^d \rarr \R$ and $g_{\min} : \R^d \rarr \R$ as
	\begin{align}
		g_{\max}(y) &\triangleq \max\{g_k(y) : k \in \{1,\ldots,m \} \} \\
		g_{\min}(y) & \triangleq \min\{g_k(y) : k \in \{1,\ldots,m\} \}
	\end{align}
	Then all of the following statements hold:
	\begin{enumerate}
		\item $f_{\max}$ and $f_{\min}$ are locally Lipschitz at $x$
		\item Let $I_{\max}(x)$ denote the set of indices $k$ for which $g_k(x) = g_{\max}(x)$. Then the function $g_{\max}$ is locally Lipschitz at $x$, and
	\begin{align}
	\label{eq:gmax}
		\partial g_{\max} \subseteq \text{co} \bigcup \{\partial g_i(x) : i \in I_{\max}(x) \}.
	\end{align}
	Furthermore, if $g_i$ is regular\footnote{The precise definition of \emph{regular functions} can be found in \cite[Defn. 2.3.4]{clarke1990optimization} and \cite{cortes2008discontinuous}. Notably, all convex functions are regular \cite[Prop. 2.3.6]{clarke1990optimization}.} at $x$ for all $i \in I_{\max}(x)$, then equality holds in \eqref{eq:gmax} and $g_{\max}$ is regular at $x$.
	\item Let $I_{\min}(x)$ denote the set of indices $k$ for which $g_k(x) = g_{\min}(x)$. Then the function $g_{\min}$ is locally Lipschitz at $x$, and
	\begin{align}
	\label{eq:gmin}
		\partial g_{\min} \subseteq \text{co} \bigcup \{\partial g_i(x) : i \in I_{\min}(x) \}.
	\end{align}
	Furthermore, if $-g_i$ is regular at $x$ for all $i \in I_{\min}(x)$, then equality holds in \eqref{eq:gmax} and $-g_{\min}$ is regular at $x$.
	\end{enumerate}
	
\end{proposition}

The \emph{set-valued Lie derivative} is used to analyze the stability of differential inclusions:

\begin{define}[\cite{cortes2008discontinuous,franceschelli2017finite}]
\label{def:setvaluedLie}
	Given a locally Lipschitz function $V : \R^d \rarr \R$ and a set-valued map $\mathcal{F} : \R^d \rarr \Pc(\R^d)$, the \emph{set-valued Lie derivative} $\widetilde{\Le}_{\mathcal{F}} V : \R^d \rarr \Pc(\R^d)$ of $V$ with respect to (w.r.t.) $\Fc$ at $  x$ is defined as 
	\begin{align}
		\widetilde{\Le}_\Fc V(  x) = \{a \in \R : \exists   v \in \Fc(  x) \text{ such that }   \zeta^T   v = a \nonumber \\
		\text{for all }   \zeta \in \partial V(  x) \} 	
	\end{align}
\end{define}

Given a locally Lipschitz and regular function $f$ and a Caratheodory solution $x(t)$ of \eqref{eq:examplediffincl}, the following result describes properties of the time derivative of the composition $f(x(t))$.

\begin{proposition}[\cite{cortes2008discontinuous,
franceschelli2017finite}]
\label{prop:Vdot}
	Let $  x : [0, t_1] \rarr \R^d$ be a solution of the differential inclusion \eqref{eq:examplediffincl} with $\Fc(\cdot)$ satisfying the hypotheses of Proposition \ref{prop:S2}, and let $h: \R^d \rarr \R$ be locally Lipschitz and regular. Then the composition $t \mapsto h(  x(t))$ is differentiable at almost all $t \in [t_0,t_1]$, and the derivative of $t \mapsto h(  x(t))$ satisfies
	\begin{align}
		\frac{d}{dt}(h(  x(t))) \in \widetilde{\Le}_\Fc h(  x(t)) 
	\end{align}
	for almost every $t \in [0,t_1]$.
\end{proposition}

Lastly, the following result will be used to demonstrate finite-time convergence.

\begin{theorem}[\cite{franceschelli2017finite}]
\label{thm:finitetime}
	Let $\mathcal{M} = \text{span}( \bm 1)$. Consider a scalar function $V(x) : \R^d \rarr \R$ with $V(x) = 0$ for all $x \in \mathcal{M}$ and $V(x) > 0$ for all $x \in \R^d \backslash \mathcal{M}$. Let $x : \R \rarr \R^d$ and $V(x(t))$ be absolutely continuous on $[t_0,\infty)$ with $d/dt(V(x(t))) \leq -\eps < 0$ almost everywhere on $\{t : x(t) \notin \mathcal{M} \}$. Then $V(x(t))$ converges to 0 in finite time, implying that $x(t)$ reaches the subspace $\mathcal{M}$ in finite time. 
\end{theorem}


\section{Main Results}
\label{sec:main}

The first Lemma of this paper describes a differential inclusion for the total system in \eqref{eq:totalsystem} under the controller \eqref{eq:uFTRC} and demonstrates that it satisfies all the conditions of Proposition \ref{prop:S2}. This will guarantee existence of solutions despite the discontinuous nature of \eqref{eq:uFTRC} and the possibly discontinuous nature of the adversarial signals.


\begin{lemma}
\label{lem:formofG}
Consider the system \eqref{eq:totalsystem} where all normally behaving agents apply the FTRC Protocol (Algorithm \ref{alg:FTRC}). 
Then the dynamics of the system \eqref{eq:totalsystem} 
satisfy the differential inclusion
\begin{align}
\label{eq:realdiffincl}
\dot{x}_{\N}(t) \in G[\fN](\xN(t)),
\end{align}
where
\begin{align}
G[\fN](\xN) = \cvxc\brc{-\alpha \bm 1, \alpha \bm 1}.
\end{align}
Furthermore, $G[\fN](\xN)$ satisfies all the hypotheses of Proposition \ref{prop:S2} {\rv and is locally Lipschitz} for all $\xN \in \R^{|\N|}$ and for all $t \geq 0$.
\end{lemma}

\begin{proof}
See Appendix Section \ref{pf:Lemma1}.
\end{proof}


We will next characterize the functions $M(\cdot)$, $m(\cdot)$, and $V(\cdot)$. These results will be necessary to demonstrate that FTRC is achieved by the system of normal agents.

\begin{lemma}
\label{lem:regLipcont}
Let the functions $ M : \RN \rarr \R$, $m : \RN \rarr \R$, and $V : \RN \rarr \R$ be defined as in \eqref{eq:Lyapunov}.
Then $M(\cdot)$, $(-m(\cdot))$, and $V(\cdot)$ are all regular, locally Lipschitz, and absolutely continuous on $\RN$.
\end{lemma}

\begin{proof}
See Appendix Section \ref{pf:Lemma2}.
\end{proof}


We next derive the Clarke generalized gradients for $M(\cdot)$ and $m(\cdot)$, which are defined in \eqref{eq:Lyapunov}. 

\begin{lemma}
\label{lem:partialms}
	Let $M : \R^{|\N|} \rarr \R$ and $m: \R^{|\N|} \rarr \R$ be defined as in \eqref{eq:Lyapunov}.
	Let $\{\N_1, \ldots,\N_{|\N|}\}$ be the indices of the normal agents, with $\N_i$ being the index of the $i$th agent in $\N$.
	The Clarke generalized gradients $\partial M$ and $\partial m$ are
	
	\begin{align}
		\partial M(\xN) &= \cvx \bigcup \brc{ e^i : \N_i \in S_M}, \label{eq:partialM} \\
		\partial m(\xN) &= \cvx \bigcup \brc{  e^i : \N_i \in S_m}. \label{eq:partialm}
	\end{align}
\end{lemma}

\begin{proof}
See Appendix Section \ref{pf:Lemma3}.
\end{proof}

The next theorem proves that $m(\xN(t))$ is nondecreasing on the interval $t \in [0, t_1)$ and that $M(\xN(t))$ is nonincreasing on the interval $t \in [0, t_1)$, where $[0,t_1)$ is the interval on which $\xN(t)$ is a solution to \eqref{eq:realdiffincl}. This will imply that the states of all agents remain within the invariant set $[m(\xN(0)), M(\xN(0))]$ for all $t \geq 0$.

\begin{theorem}
\label{thm:derivatives}
Consider a digraph $\D = \{\V, \E\}$ with the system dynamics \eqref{eq:realdiffincl} under the {\rv FTRC Protocol in Algorithm \ref{alg:FTRC}}. Suppose that $\A$ is an $F$-local model and that $\D$ is $(2F+1)$-robust. Let $m(\xN(t))$ and $M(\xN(t))$ be defined as in \eqref{eq:Lyapunov}. Then the derivatives $\frac{d}{dt} (M(\xN(t)))$ and $\frac{d}{dt} (m(\xN(t)))$ exist at almost all $t \in [0, t_1)$ and satisfy 
\begin{align}
\frac{d}{dt} (M(\xN(t))) &\in [-\alpha, 0], \\
\frac{d}{dt} (m(\xN(t))) &\in [0, \alpha],
\end{align}
at almost all $t \in [0, t_1)$.
\end{theorem}

\begin{proof}
See Appendix Section \ref{pf:Theorem3}.
\end{proof}

The preceding Lemma demonstrates that $M(\xN(t))$ is nonincreasing and $m(\xN(t))$ is nondecreasing for all $t \geq 0$, and therefore all agents' states remain within the interval $[m(\xN(0)), M(\xN(0))]$ for all $t \geq 0$. This implies that the hyperrectangle $P(0) \subset \R^{|\N|}$ defined as
\begin{align}
\label{eq:Pset}
P(0) = \bmx{[m(\xN(0)), M(\xN(0))] \\ \vdots \\ [m(\xN(0)), M(\xN(0))}
\end{align}
 is invariant for all $t \geq 0$, which is precisely condition (i) of {\rv Finite-Time Resilient Consensus} (Definition \ref{def:FTRC}).
 
The next result demonstrates that the time derivative of the composition $V(\xN(t))$, wherever it exists, is upper bounded by $-\alpha$ when $\xN(t)$ is not in $\text{span}(\bm 1)$.

\begin{theorem}
\label{thm:Vless}
Let $V(\cdot)$ be defined as in \eqref{eq:Lyapunov}.
Under the conditions of Theorem \ref{thm:derivatives}, the derivative $\frac{d}{dt}V(\xN(t))$ exists at almost all $t \in [0,t_1)$. Furthermore, for all $\xN(t) \notin \text{span}(\bm 1)$, the derivative of $V(\xN(t))$ satisfies
\begin{align}
\label{eq:Vless}
	\frac{d}{dt} (V(\xN(t))) \leq -\alpha < 0
\end{align}
at almost all $t \in [0,t_1)$.
\end{theorem}

\begin{proof}
See Appendix Section \ref{pf:Theorem4}.
\end{proof}

Our final theorem completes the paper by showing that FTRC is achieved by the system of normal agents.
{\rv
In particular, this theorem demonstrates that solutions to the trajectories of the normal agents exist on the time interval $t \in [0,\infty)$, and that there exists a time $T \geq 0$ such that $\xN(t) \in  \text{span}(\bm 1)$ for all $t \geq T$.
}

\begin{theorem}
\label{thm:final}
Consider a digraph $\D = \{\V, \E\}$ with the system dynamics \eqref{eq:realdiffincl} under the {\rv FTRC Protocol in Algorithm \ref{alg:FTRC}}. Suppose that $\A$ is an $F$-local model and that $\D$ is $(2F+1)$-robust. Then the normal agents achieve FTRC as described in Definition \ref{def:FTRC}.
\end{theorem}

\begin{proof}
See Appendix Section \ref{pf:Theorem5}.
\end{proof}

\section{Simulations}
\label{sec:simulations}

Our simulations are for a system of $n = 15$ agents. The underlying communication graph is a $k$-circulant digraph with $k = 11$, which can be shown to be at least $6$-robust using results from \cite{usevitch2017circulant}. The highest (integer) value of $F$ for which we can infer the graph is $(2F+1)$-robust is therefore $F = 2$. Each agent's initial state $x_i(0) \in \R$, $i \in \V$ is a random value on the interval $[0,50]$. Two agents are chosen at random to be adversaries, resulting in $\A = \{2,13\}$. We emphasize that \emph{the normally-behaving agents have no knowledge as to whether their in-neighbors are adversarial or normal}. The adversarial agents are \emph{malicious} \cite{leblanc2013resilient}, meaning each adversary updates its state according to some arbitrary function of time but sends the same state value to all of its out-neighbors. All other agents are normal and apply the FTRC-P from Algorithm \ref{alg:FTRC} with $\alpha = 10$. The function $g : \R \rarr \R$ in \eqref{eq:uFTRC} is chosen to be $g(x) = (1/10)x^3 + (1/1000)x^5 + (1/10000)x^7$, which can be verified to be a strictly increasing function.
Figure \ref{fig:sim1} shows the results of this first simulation, with malicious agents represented by red dotted lines and normally-behaving agents represented by solid colored lines. The normal agents achieve consensus in a finite amount of time despite the influence of the adversarial agents.

\begin{figure}
\centering
\includegraphics[width=.8\columnwidth]{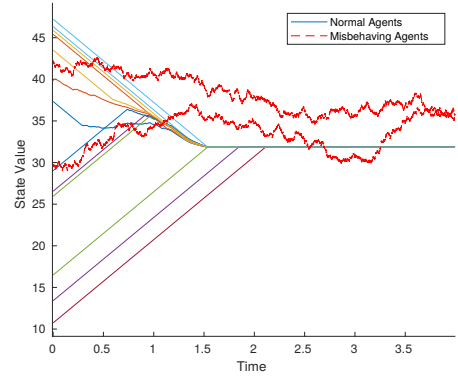}
\caption{Simulation of a network of 15 agents appling the FTRC-P. The dotted red lines represent the adversarial agents.}
\label{fig:sim1}
\end{figure}

\section{Conclusion}
\label{sec:conclusion}

In this paper we presented conditions under which finite-time convergence of normally-behaving agents in the presence of discontinuous, nonlinear adversarial signals is guaranteed. Future work will extend the use of discontinuous systems theory to other resilient continuous-time control objectives.




\section{APPENDIX}
\label{sec:appendix}

\subsection{Proof of Lemma 1}
\label{pf:Lemma1}

By the definition of the $\sign(\cdot)$ function, observe that for all $i \in \N$ we have $u_i \in \{-\alpha, 0, \alpha\}$. Note that this holds for all possible adversarial signals $x_{\A}^{\N}$ defined in \eqref{eq:alladversaries}. Therefore $\dot{x}_i(t) \in [-\alpha, \alpha]$ for all $i \in \N$, implying that $\dot{ x}_{\N}(t) \in \cvxc\brc{-\alpha \bm 1, \alpha \bm 1} = G[\fN] (\xN)$ for all $\xN \in \R^{|\N|}$.

Next, we show that $G[\fN] (\xN)$ satisfies all the hypotheses of Proposition \ref{prop:S2}. Note that $G[\fN] (\xN)$ takes nonempty, compact, and convex values. Since $G[\fN](\xN)$ is time-invariant and equal to the Cartesian product of intervals $[-\alpha, \alpha] \times \ldots \times [-\alpha,\alpha]$, it is measurable for all $x \in \R^{|\N|}$ and for all $t \geq 0$. 
To show local boundedness note that for all $\xN \in \R^{|\N|}$, for all $t \geq 0$, and for all $v \in \G[f](\xN)$ we have $\nrm{v}_2 = \pth{\sum_{i = 1}^{|\N|} |v_i|^2}^{(1/2)} \leq \pth{\sum_{i = 1}^{|\N|} |\alpha|^2}^{(1/2)} = \sqrt{|\N|} \alpha$. Letting $\gamma(t) = \sqrt{|\N|} \alpha$, it follows that for all $(t, x) \in [0,\infty)\times \R^{|\N|}$ and for all $\eps,\delta > 0$ we have $\nrm{v}_2 \leq \gamma(s)$ for all $v \in \G[f](\xN)$, for all $s \in [t,t+\delta]$, and for all $y \in B(x,\eps)$.

{\rv
Finally, $G[f_\N](x_\N)$ can be shown to be locally Lipschitz by noting that since $G[f_\N](x_\N) $ is constant for all $x \in \R^{|\N|}$, it holds that for all $x_{|\N|} \in \R^{\N}$ there exists $L > 0$, $\eps > 0$ such that $G[f_\N](y) = \cvxc\brc{-\alpha \bm 1, \alpha \bm 1} \subseteq \cvxc\brc{-\alpha \bm 1, \alpha \bm 1} + L \nrm{y - z}\bar{B}(0,1) = G[f_\N](z) + L \nrm{y - z}\bar{B}(0,1) $ for all $y,z \in B(x_{\N}, \eps)$. Since local Lipschitzness of $G[f_\N](x_\N)$ implies upper semicontinuity of $G[f_\N](x_\N)$ \cite{cortes2008discontinuous}, $G[f_\N](x_\N)$ therefore satisfies all the hypotheses of Proposition \ref{prop:S2}.
}

\subsection{Proof of Lemma 2}
\label{pf:Lemma2}

Recall that $e^i$ is the $i$th column of the $|\N| \times |\N|$ identity matrix. Observe that $M(\xN)$ is the pointwise maximum over the functions $(e^i)^T \xN$ for $i \in \N$, which are all locally Lipschitz on $\RN$. By Proposition \ref{prop:maxfunc}, $M(\xN)$ is therefore locally Lipschitz on $\RN$. In addition, each function $(e^i)^T \xN$ is affine, and therefore convex and regular on $\RN$. Since for all possible indices $i \in \{1,\ldots,\N\}$ the functions $(e^i)^T \xN$ are regular on $\RN$, by Proposition \ref{prop:maxfunc} $M(\xN)$ is regular on $\RN$.

Similarly, $m(\xN)$ is the pointwise minimum over the functions $(e^i)^T \xN$ for $i \in \{1,\ldots,\N\}$, which are all locally Lipschitz on $\RN$. By Proposition \ref{prop:maxfunc}, $m(\xN)$ is therefore locally Lipschitz on $\RN$.
Since each $(e^i)^T \xN$ is affine, each function $-(e^i)^T \xN$ is also affine and therefore convex and regular on $\RN$. Therefore by Proposition \ref{prop:maxfunc} the function $(-m(\xN))$ is regular on $\RN$. 

Since $V(\xN)$ is equal to the sum of two locally Lipschitz and regular functions, it holds that $V(\xN)$ is also locally Lipschitz and regular \cite{cortes2008discontinuous}. Finally, every locally Lipschitz function on $\RN$ is absolutely continuous on $\RN$ \cite{cortes2008discontinuous}, which implies that $M(\xN)$, $(-m(\xN))$, and $V(\xN)$ are all absolutely continuous.

\subsection{Proof of Lemma 3}
\label{pf:Lemma3}

By Lemma \ref{lem:regLipcont}, $M(\xN)$ is the pointwise maximum over the functions $(e^i)^T \xN$ for $i \in \{1,\ldots,|\N| \}$, which are all locally Lipschitz and regular on $\RN$. Furthermore, each function $(e^i)^T \xN$ is continuously differentiable at all $\xN \in \RN$, implying that $\partial((e^i)^T \xN) = \nabla((e^i)^T \xN) = e^i$ \cite{cortes2008discontinuous}. By Proposition \ref{prop:maxfunc}, we therefore have
\begin{align}
\partial M(\xN) = \cvx \bigcup \brc{e^i : i \in I_{\max}(\xN)}, \label{eq:partialpartialM}
\end{align}
where $I_{\max}(\xN)$ denotes the indices $j$ such that $(e^j)^T \xN = x_{\N_j} = M(\xN)$ (recall from \eqref{eq:normalvector} that $(e^j)^T \xN = x_{\N_j}$, where $\N_j$ is the index of the $j$th normal agent in $\N$). By equation \eqref{eq:Smsets}, the set of indices $\N_j$ such that $x_{\N_j} = M(\xN)$ is precisely $S_M(\xN)$, which by substitution into \eqref{eq:partialpartialM} yields \eqref{eq:partialM}.

Similar arguments can be used to derive $\partial m(\xN)$. The function $m(\xN)$ is the pointwise minimum over the functions $(e^i)^T \xN$ for $i \in \{1,\ldots,|\N| \}$ which are all locally Lipschitz, regular, and continuously differentiable on $\RN$. Observe that the functions $-(e^i)^T \xN$ are also locally Lipschitz, regular, and continuously differentiable on $\RN$. By Proposition \ref{prop:maxfunc}, we therefore have
\begin{align}
\partial m(\xN) = \cvx \bigcup \brc{e^i : i \in I_{\min}(\xN)}, \label{eq:partialpartialm}
\end{align}
where $I_{\min}(\xN)$ denotes the indices $j$ such that $(e^j)^T \xN = x_{\N_j} = m(\xN)$. By Definition \ref{eq:Smsets}, the set of indices $\N_j$ such that $x_{\N_j} = m(\xN)$ is precisely $S_m(\xN)$, which by substitution into \eqref{eq:partialpartialm} yields \eqref{eq:partialm}.
{\rv
As a final note, observe that since $m(\cdot)$ is locally Lipschitz by Lemma \ref{lem:regLipcont}, by the Dilation Rule \cite{cortes2008discontinuous} we can derive $\partial(-m(x_\N)) =  \partial ((-1)m(x_\N)) = - (\partial m(x_\N))$.
}

\subsection{Proof of Theorem 3}
\label{pf:Theorem3}

{\rv We will first need the following Lemma for the proof of Theorem 3.}

\begin{lemma}
\label{lem:convex}
Let $ q \in \R^m$ and let $\Theta = \{ \theta \in \R^m :  \theta \ggeq \bm 0,\ \bm 1^T  \theta = 1\}$. Let $a \in \R$. Then $ \theta^T  q = a$ for all $ \theta \in \Theta$ if and only if $ q = a \bm 1$. 
\end{lemma}

\begin{proof}
\emph{Necessity:} If $ q = a \bm 1$, then for all $ \theta \in \Theta$ we have $ \theta^T  q =  q^T  \theta = a(\bm 1^T \theta) = a$.

\emph{Sufficiency:} We prove the contrapositive, i.e. $ q \neq a \bm 1$ implies there exists $ \theta^* \in \Theta$ such that $( \theta^*)^T  q \neq a$. If $ q \neq a \bm 1$ then there exists $j \in \{1,\ldots,m\}$ such that $q_j \neq a$. Choose $ \theta^* =  e_j$, where $ e_j$ is the $j$th column of the identity matrix. Clearly, we then have $ \theta_* \ggeq \bm 0$ and $\bm 1^T  \theta^* = 1$, implying $ \theta^* \in \Theta$. Then $( \theta^*)^T  q =  e_j^T  q = q_j \neq a$. 
\end{proof}

{\rv We now give the proof of Theorem 3.}
Where possible, we abbreviate $\xN(t)$ to $\xN$ for brevity.
By Lemma \ref{lem:formofG}, solutions $\xN(t)$ to the differential inclusion \eqref{eq:realdiffincl} are guaranteed.
By Lemma \ref{lem:regLipcont}, the functions $M(\cdot)$ and $(-m(\cdot))$ are both locally Lipschitz and regular on $\RN$.
Therefore by Proposition \ref{prop:Vdot}, the compositions $M(\xN(t))$ and $(-m(\xN(t)))$ are differentiable at almost all $t \in [0, t_1)$. 
In addition, by Proposition \ref{prop:Vdot} we have $\frac{d}{dt} M(\xN) \in \Lt_{G} M(\xN)$ and $\frac{d}{dt} (-m(\xN)) \in \Lt_{G} (-m(\xN))$ at almost all $t \in [0, t_1)$, where $\Lt_{G} M(\xN)$ and $\Lt_{G} (-m(\xN))$ represents the set-valued Lie derivatives of $M(\xN)$ and $(-m(\xN))$, respectively. 
The next part of the proof focuses on characterizing $\Lt_{G} M(\xN)$ and $\Lt_{G} (-m(\xN))$, from which we derive the range of possible values for $\frac{d}{dt} M(\xN)$ and $\frac{d}{dt} m(\xN)$.

We first consider $\Lt_{G} M(\xN)$. By definition, 
\begin{align}
\Lt_{G} M(\xN) = &\{a \in \R : \exists  v \in G[\fN](\xN) \text{ such that } \nonumber\\ 
& \hspace{2em} z^T v = a\ \forall  z \in \partial M(\xN)\}
\end{align}
Define $E_M$ as a matrix with columns $ e^i$ such that $\N_i \in S_M$.\footnote{Recall that $\N_i$ is defined immediately after Eq. \eqref{eq:normalvector}.} By the definition of $\partial M(\xN)$ from Lemma \ref{lem:partialms}, each $z \in \partial M(\xN)$ can be written as the convex combination $z = E \theta$, where $\theta \in \R^{|S_M|}$, $\theta \ggeq \bm 0$ and $\bm 1^T \theta = 1$. 
It therefore holds that $a \in \Lt_{G} M(\xN)$ if and only if there exists a $v \in G[\fN](\xN)$ such that $z^T v = ( \theta^T E_M^T)v = \theta^T (E_M^T v) = a$ for all $ \theta \ggeq \bm 0$, $\bm 1^T  \theta = 1$. Lemma \ref{lem:convex} in the Appendix proves that this holds if and only if $E_M^T  v = a \bm 1$. Recall that $E_M$ is composed of the columns $e^i$ such that $\N_i \in S_M$. By the form of $G[\fN](\xN)$, choosing any $ v \in G[\fN](\xN)$ with $v_i = a \in [-\alpha,\alpha]$ for all $i$ such that $\N_i \in S_M$ yields $E_M^T v = a \bm 1$, and therefore $\frac{d}{dt} M(\xN) \in \Lt_{G} M(\xN) = [-\alpha,\alpha]$.

We can further restrict the range of values for $\frac{d}{dt} M(\xN)$ to the range $[-\alpha,0]$ by considering the form of \eqref{eq:uFTRC}. 
We prove this by contradiction. Suppose there exists a $t \geq 0$ such that $\frac{d}{dt} M(\xN(t)) > 0$. This implies that there exists $t \geq 0$ and $\N_{i'} \in S_M(t)$ such that $u_{\N_{i'}}(t) = \alpha\, \sign\pth{\sum_{\J_{\N_{i'}} \backslash \Rc_{\N_{i'}}[t]} g(x_j^{\N_{i'}}(t)) - g(x_{\N_{i'}}(t)) } > 0$. 
However, for all $\N_i \in S_M$ all normal in-neighbors $j \in \V_i(t)$ have state values less than or equal to $x_{\N_i}(t)$ by the definition of $S_M$. Since $g(\cdot)$ is strictly increasing, we have $g(x_j^{\N_i}) - g(x_{\N_i}) \leq 0$ for all normal in-neighbors $j \in \V_{\N_i} \cap \N$. In addition, since $\A$ is $F$-local, any adversarial signals satisfying $g(x_k^{\N_i}(t)) > g(x_{\N_i}(t))$ for $k \in (\V_{\N_i} \cap \A)$ are filtered out by Algorithm \ref{alg:FTRC}. Therefore we must have $\sum_{\J_{\N_i} \backslash \Rc_{\N_i}[t]} g(x_j^{\N_i}(t)) - g(x_{\N_i}(t)) \leq 0$ for all ${\N_i} \in S_M$, which implies that $u_{{\N_i}}(t) = (\alpha) \sign\pth{\sum_{\J_{\N_i} \backslash \Rc_{\N_i}[t]} g(x_j^{\N_i}(t)) - g(x_{\N_i}(t)) } \leq 0$ for all ${\N_i} \in S_M$. This contradicts the assumption that there exists an ${\N_{i'}} \in S_M$ with $u_{\N_{i'}}(t) > 0$. Therefore $\frac{d}{dt} M(\xN) \leq 0$ wherever it exists, which yields $\frac{d}{dt} M(\xN) \in [-\alpha, 0]$.

The preceding logic can be repeated to demonstrate that {\rv $\frac{d}{dt} (-m(\xN)) \in [-\alpha,0]$ wherever this derivative exists, from which we can conclude that $\frac{d}{dt} m(\xN) \in [0,\alpha]$.}

\subsection{Proof of Theorem 4}
\label{pf:Theorem4}

{\rv We will need the following Lemma for the proof of Theorem 4.}
 
\begin{lemma}
\label{eq:samederivatives}
Under the conditions of Theorem \ref{thm:derivatives}, if $\frac{d}{dt} M(\xN(t))$ exists at $t \geq 0$ then $u_{i_1}(t) = u_{i_2}(t)$ for all $i_1, i_2 \in S_M$. Similarly, if $\frac{d}{dt} m(\xN(t))$ exists at $t \geq 0$, then $u_{j_1}(t) = u_{j_2}(t)$ for all $j_1, j_2 \in S_m$.
\end{lemma}

\begin{proof}
We prove the contrapositive. If at some $t \geq 0$ there exists $i_1, i_2 \in S_M$ such that $u_{i_1}(t) \neq u_{i_2}(t)$, then by \eqref{eq:dynamics} $\dot{x}_{i_1}(t) \neq \dot{x}_{i_2}(t)$. Since $M(\xN(t))$ is the pointwise maximum $\max_{i \in \N} x_i(t)$ and $x_{i_1}(t) = x_{j_1}(t) = M(\xN(t))$ by definition of $S_M$, the derivative $\frac{d}{dt} M(\xN(t))$ is therefore undefined at $t$. Similar arguments demonstrate the same result for $\frac{d}{dt} m(\xN(t))$.
\end{proof} 

{\rv We now give the proof of Theorem 4.}
Where possible, we abbreviate $\xN(t)$ to $\xN$ for brevity.
By definition, $V(\xN(t)) = M(\xN(t)) - m(\xN(t))$ which implies $\frac{d}{dt} V(\xN(t)) = \frac{d}{dt} M(\xN(t)) - \frac{d}{dt} m(\xN(t))$. Since $\frac{d}{dt} M(\xN(t))$ and $\frac{d}{dt} m(\xN(t))$ exist at almost all $t \in [0,t_1)$, $\frac{d}{dt} V(\xN(t))$ exists at almost all $t \in [0,t_1)$.

Next, we show that for all $\xN \notin \text{span}(\bm 1)$, there exists an agent $i \in (S_M \cup S_m)$ such that either $u_i(t) = -\alpha$ or $u_i(t) = \alpha$. Observe that $\xN \notin \text{span}(\bm 1)$ implies that $S_M$ and $S_m$ are nonempty and disjoint. By the definition of $(2F+1)$-robustness (Definition \ref{def:rrobust}), at least one of the sets $S_M$, $S_m$ is $(2F+1)$-reachable. Without loss of generality, suppose $S_M$ is $(2F+1)$-reachable. Then there exists $i \in S_M$ with $|\V_i \backslash S_M| \geq 2F+1$. By the FTRC-P, agent $i$ will filter out at most $2F$ values. Since $i \in S_M$, any normal values received by $i$ will be less than or equal to $g(x_i(t))$. Since $\A$ is $F$-local, any adversarial values greater than $g(x_i(t))$ will be filtered out as per the FTRC-P. This implies that agent $i$ will \emph{not} filter out at least one value $g(x_j^i(t)) < g(x_i(t))$, and that $\sum_{\J_i \backslash \Rc_i[t]} g(x_j^i(t)) - g(x_i(t)) < 0$. Therefore $u_i(t) = -\alpha$. Similar arguments can be used to show that if $S_m$ is $(2F+1)$-reachable, there exists $i \in S_m$ with $u_i(t) = \alpha$. 

Consider any $t\geq t_0$ such that $\xN(t) \notin \text{span}(\bm 1)$ and $\frac{d}{dt} V(\xN)$ exists. The existence of $\frac{d}{dt} V(\xN(t))$ implies that both $\frac{d}{dt} M(\xN)$ and $\frac{d}{dt} m(\xN)$ exist. Since $\xN(t) \notin \text{span}(\bm 1)$, by prior arguments there either exists a $i_M \in S_M$ with $u_{i_M}(t) = -\alpha$ or an $i_m \in S_m$ with $u_{i_m}(t) = \alpha$. We consider each case separately.

\emph{Case 1:} Suppose there there exists an $i_M$ with $u_{i_M}(t) = -\alpha$. Recall that we are considering any $t \geq 0$ such that $\xN(t) \notin \text{span}(\bm 1)$ and $\frac{d}{dt} V(\xN)$ exists, implying that $\frac{d}{dt} M(\xN)$ exists. Since $\frac{d}{dt} M(\xN)$ exists at $t$, then by Lemma \ref{eq:samederivatives} we have $u_j(t) = -\alpha$ at $t$ for all $j \in S_M$. This implies that $\frac{d}{dt} M(\xN) = -\alpha$. Since $\frac{d}{dt} m(\xN)$ also exists at our chosen $t$ and $m(\xN) \in [0,\alpha]$, we have $\frac{d}{dt}V(\xN) \leq -\alpha < 0$.

\emph{Case 2:} Suppose there there exists an $i_m$ with $u_{i_m}(t) = \alpha$ Since $\frac{d}{dt} m(\xN)$ exists at our choice of $t$, then by Lemma \ref{eq:samederivatives} we have $u_j(t) = \alpha$ at $t$ for all $j \in S_M$. This implies that $\frac{d}{dt} m(\xN) = \alpha$. Since $\frac{d}{dt} M(\xN)$ also exists at our chosen $t$ and $M(\xN) \in [-\alpha,0]$, we have $\frac{d}{dt}V(\xN) \leq -\alpha < 0$.

Since in each case we have $\frac{d}{dt} V(\xN) \leq -\alpha < 0$, for all $\xN(t) \notin \text{span}(\bm 1)$ the equation \eqref{eq:Vless} holds at almost all $t \in [0,t_1)$.

\subsection{Proof of Theorem 5}
\label{pf:Theorem5}

By Theorem \ref{thm:derivatives}, all normal agents remain within the invariant set $P(0)$ defined in \eqref{eq:Pset}, satisfying condition (i) of FTRC. By Theorem \ref{thm:Vless}, condition (ii) of FTRC is satisfied. 
To show that condition (iii) of FTRC is satisfied, observe that by Lemma \ref{lem:regLipcont} $V(\cdot)$ is locally Lipschitz on $\RN$. Since Caratheodory solutions $\xN(t)$ of \eqref{eq:realdiffincl} are absolutely continuous, the composition $V(\xN(t))$ is therefore absolutely continuous \cite[Appendix B]{franceschelli2017finite}. 
By Lemma \ref{lem:formofG} $G[f](\xN)$ satisfies the hypotheses of Proposition \ref{prop:S2} for all $\xN \in \RN$ and for all $t \geq 0$, implying that these hypotheses are satisfied for the compact set $Q = \cvxc(P(0) + B(0,\eps))$ for some $\eps > 0$ (where addition is in terms of the Minkowski sum). Since $P(0)$ is an invariant set and $P(0)$ does not intersect the boundary of $Q$, no solution $\xN(t)$ will reach the boundary of $Q$ for all $t \geq 0$. 
Consider any domain $D(t_1') = [-\delta,t_1'] \times Q$ for $\delta,t_1' > 0$. Each domain $D(t_1')$ is therefore compact.
By Theorem \ref{thm:Filippov} this implies that all solutions $\xN(t)$ to \eqref{eq:realdiffincl} exist on $t \in [0,t_1')$ for any $t_1' > 0$, which implies that all solutions $\xN(t)$ to \eqref{eq:realdiffincl} exist on $t \in [0,\infty)$.
By Theorems \ref{thm:finitetime} and \ref{thm:Vless} $V(\xN(t))$ converges to $\text{span}(\bm 1)$ in finite time, implying that $\xN(t)$ reaches consensus in finite time and condition (iii) of FTRC is satisfied.
{\rv
Since by Theorem \ref{thm:Vless} we have $\frac{d}{dt}(V(\xN(t))) \leq -\alpha$ at almost all $t \in [0,\infty)$, the time of convergence satisfies $T(\xN(0)) = \frac{1}{\alpha}V(\xN(0))$. 
}

{\rv
\subsection{Discussion of Assumption \ref{assume:measurable}}
\label{sec:discussionassume}

In this section we discuss further the implications of Assumption \ref{assume:measurable}. 
Specifically, we consider the possibility of the adversaries sending signals which are not Lebesgue measurable.
To give a simple example of non-Lebesgue-measurable function, the indicator function $\bm 1_S : \R \rarr \R$ defined as
\begin{align*}
\bm 1_S(x) = \begin{cases}
1 & \text{if } x \in S \\
0 & \text{otherwise}
\end{cases}
\end{align*}
is not Lebesgue measurable if the subset $S \subset \R$ is not Lebesgue measurable.
Note that by definition of measurability, the existence of a non-Lebesgue-measurable mapping from $\R$ to $\R$ implies the existence of a subset of $\R$ which is not Lebesgue measurable. Contrapositively, the nonexistence of non-Lebesgue-measurable subsets of $\R$ implies that all functions mapping $\R$ to $\R$ are Lebesgue measurable.

There are at least two schools of thought on this point. If one assumes that the axiom of choice holds, then the axiom of choice can be used to demonstrate the existence of subsets of $\R$ which are not Lebesgue measurable (e.g. Vitali sets \cite{cichon1993sets}). 
However, the Solovay model \cite{solovay1970model} demonstrated that the existence of a non-Lebesgue-measurable subset of $\R$ cannot be proven without using the axiom of choice. Under the Solovay model, which does not assume the axiom of choice but instead assumes the existence of an inaccessible cardinal, \emph{all} subsets of $\R$ are Lebesgue measurable. 

The question of whether the adversaries can send non-Lebesgue-measurable signals therefore hinges upon which assumptions are made about the axiom of choice and the existence of an inaccessible cardinal. A full discussion of the merits of each approach is completely beyond the scope of this paper, and so we conclude by simply asserting that the results of this paper hold under Assumption \ref{assume:measurable}, i.e. when all adversarial signals are Lebesgue measurable.
}

\bibliographystyle{IEEEtran}
\bibliography{ACC2020}

\end{document}